\documentclass[11pt]{article}
\usepackage[margin=1.05in]{geometry}
\usepackage{amsmath,amssymb,amsthm,mathtools}
\usepackage{microtype}
\usepackage{tikz}
\usepackage[hidelinks]{hyperref}
\newtheorem{theorem}{Theorem}[section]
\newtheorem{proposition}[theorem]{Proposition}
\newtheorem{lemma}[theorem]{Lemma}
\newtheorem{corollary}[theorem]{Corollary}
\theoremstyle{remark}
\newcommand{\E}{\mathbb E}
\newcommand{\Prob}{\mathbb P}
\newcommand{\R}{\mathbb R}
\newcommand{\Z}{\mathbb Z}
\newcommand{\Var}{\operatorname{Var}}

\newcommand{\Conf}{\mathcal N}
\newcommand{\TV}{\operatorname{TV}}
\title{Uniform displacement bounds and Gibbs limits\\for periodic one-dimensional Riesz gases}
\author{Yan Ru Pei}
\date{September 14, 2026}
\begin{document}
\maketitle
\begin{abstract}
For the neutral periodic one-dimensional Riesz gas with pair potential
locally $-|x|^a$, $0<a<1$, we prove a particle-displacement variance
bound of order $\beta^{-1}$, uniformly in the number of particles.
Log-concavity also gives exponential displacement tails.
Every stationary periodic thermodynamic limit is simple, has intensity
one, and admits a stationary ordered matching to the unit lattice with
the same bounds. Each limit satisfies the canonical Gibbs equations
for the full-line interaction, with an ordinary symmetric spatial
principal value for the exterior potential.
The matching implies uniformly bounded interval number variance and
a positive limiting second moment of the reciprocal-lattice Fourier
average at sufficiently low temperature.
The main estimate compares the inverse random Hessian, through
deterministic electrical flows, to a transient long-range network.
\end{abstract}

\section{Introduction and main results}

The strength of positional order in a one-dimensional particle system
depends on the range of its interaction. For the neutral Coulomb plasma
with pair potential $-|x|$, ordered coordinates reduce the energy to a quadratic form;
the Brascamp--Lieb approach to the one-dimensional plasma uses this
structure to control particle displacements~\cite{BL1975}.
The negative Riesz potentials $-|x|^a$, $0<a<1$, interpolate between
logarithmic and Coulomb behavior after additive and multiplicative
normalization: $-(|x|^a-1)/a\to-\log|x|$ as $a\downarrow0$ for $x\ne0$.
They are convex away from the
collision point, but their ordered Hamiltonians have a Hessian that
depends on every particle separation. Obtaining displacement bounds
uniformly in the volume requires control of this random Hessian.

We prove such bounds for the unit-density periodic gas at every
positive inverse temperature. The proof uses the Hessian as the
Laplacian of an electrical network. Log-concavity controls expected
reciprocal conductances. A deterministic test flow then bounds the
expected effective resistance by that of the network with conductances
$|i-j|^{a-2}$. Its low-frequency eigenvalues have order
$|\theta|^{1-a}$, whose reciprocal is integrable precisely when $a>0$.
Combining this estimate with the Brascamp--Lieb variance inequality
gives the volume-independent displacement bound.

The periodic negative-Riesz gas has been studied numerically by
Lelotte~\cite{Lelotte}. Convexity and Brascamp--Lieb estimates also
play an important role in Boursier's work on circular Riesz gases
with positive exponent~\cite{Boursier}. Dereudre and
Digneaux~\cite{DD} study positional rigidity through stationary
transport to a lattice. Their negative-Riesz application uses
stationarized interval-background ensembles. Our results concern
periodic approximants; equality of the resulting infinite-volume
ensembles is a separate question. The periodic/container comparison
in~\cite{Lewin} does not identify them in the present exponent range.

The uniform estimate has two infinite-volume consequences. First,
every periodic limit admits an ordered lattice matching with square
integrable displacements. This yields bounded interval number
variance, a form of positional order related to the bounded-fluctuation
principle of Aizenman, Goldstein and Lebowitz~\cite{AGL}.
Second, the exponential displacement tails make the exterior
relative potential converge along every symmetric spatial cutoff.
A uniform comparison of the finite-circle and full-line potentials
then passes the actual circle conditional laws to canonical Gibbs
equations on the line.

\subsection{The periodic ensemble}

Fix $a\in(0,1)$ and write $\mathbb T_N=\R/(N\Z)$, where
$N\ge1$ is an integer. Define the even, continuous, zero-mean kernel
\begin{equation}\label{eq:kernel}
 g_N(r)=c_aN^a\sum_{k\in\Z\setminus\{0\}}
                  \frac{e^{2\pi ikr/N}}{|k|^{1+a}},
 \qquad
 c_a=\frac{2\Gamma(1+a)\sin(\pi a/2)}{(2\pi)^{1+a}}.
\end{equation}
The series is absolutely convergent. Its equivalent image formula is
\begin{equation}\label{eq:images}
 g_N(r)-g_N(0)=-|r|^a-
  \sum_{m\ge1}
  \bigl(|r+mN|^a+|r-mN|^a-2(mN)^a\bigr).
\end{equation}
The image series converges uniformly on compact sets; its
normalization will also be verified below. In particular the local
singular term is $-|r|^a$.

For $\beta>0$, let the law of $N$ labeled points on $\mathbb T_N$
have density
\begin{equation}\label{eq:ensemble}
 \frac1{Z_{N,\beta}}
 \exp\left\{-\beta\sum_{1\le i<j\le N}g_N(x_i-x_j)\right\}
 \prod_{i=1}^N dx_i.
\end{equation}
The normalizer is finite and positive because $g_N$ is continuous
on the compact circle. Collision sets have zero Lebesgue measure.
The induced point configuration does not depend on the labeling.
Omitting the zero Fourier mode is equivalent, up to an additive
energy constant, to including a uniform neutralizing background of
density one: $\int_{\mathbb T_N}g_N(x-y)\,dy=0$ for every $x$.

Choose one of the $N$ particles uniformly, translate it to the origin,
and order the remaining particles cyclically as
$0=Y_0<Y_1<\cdots<Y_{N-1}<N$.
We call this the particle-Palm law. It has density proportional to
$e^{-\beta H_N(Y)}$ on the ordered simplex, where
\begin{equation}\label{eq:palm-energy}
 H_N(Y)=\sum_{0\le i<j<N}g_N(Y_j-Y_i).
\end{equation}
Indeed, take the chosen labeled particle as translation coordinate
in~\eqref{eq:ensemble}. The coordinate change has Jacobian one;
the integral over the translation gives a constant factor $N$, and
ordering the other coordinates gives a constant factor $(N-1)!$.
Thus no gap-length weight occurs in this rooting convention.

\begin{theorem}[Uniform particle displacements]\label{thm:finite}
For each $a\in(0,1)$ there is $0<C_a<\infty$ such that, for every
$N\ge2$, $\beta>0$, and $0\le i<N$,
\begin{equation}\label{eq:finite}
 \E Y_i=i,\qquad \E|Y_i-i|^2\le C_a/\beta.
\end{equation}
Set $\sigma=(C_a/\beta)^{1/2}$ and $\kappa=\log(3)/4$.
For every $t\ge0$,
\begin{equation}\label{eq:exp}
 \Prob(|Y_i-i|>t)\le2e^{-\kappa t/\sigma},\qquad
 \E e^{\kappa|Y_i-i|/(2\sigma)}\le3.
\end{equation}
\end{theorem}

\subsection{Thermodynamic limits and canonical Gibbs equations}

Periodically extend the unrooted configurations with
law~\eqref{eq:ensemble} to $\R$, and denote their laws by $\mu_N$.
They are stationary and have intensity one. We use the vague
topology on the space $\Conf(\R)$ of locally finite counting measures,
allowing multiplicities in the ambient space. A stationary matching
means that the joint law of the particle configuration, the lattice,
and their matching segments is invariant under real translations.

\begin{theorem}[Ordered lattice matching]\label{thm:limit}
Fix $a\in(0,1)$ and $\beta>0$. The family $(\mu_N)$ is relatively
compact. Every weak accumulation point along $N\to\infty$ is a
simple stationary process $\xi$ of intensity one that admits a
representation
\begin{equation}\label{eq:matching}
 \xi=\sum_{j\in\Z}\delta_{j+U+q_j}.
\end{equation}
Here $U$ is uniform on $[0,1)$ and independent of the
$\Z$-stationary sequence $(q_j)$,
$1+q_{j+1}-q_j>0$ almost surely for every $j$, and
\begin{equation}\label{eq:limit-moments}
 \E q_0^2\le C_a/\beta,\qquad
 \Prob(|q_0|>t)\le2e^{-\kappa t/\sigma},\qquad
 \E e^{\kappa|q_0|/(2\sigma)}\le3.
\end{equation}
The correspondence $j+U\mapsto j+U+q_j$ is a stationary matching.
\end{theorem}

Write $g(r)=-|r|^a$. For a bounded open interval $\Lambda$, put
$\eta=\xi|_{\Lambda^c}$ and $k=\xi(\Lambda)$, and choose
$x_*\in\Lambda$. The relative exterior potential is
\begin{equation}\label{eq:potential}
 W_\eta(x)=\lim_{L\to\infty}
       \sum_{\substack{y\in\eta\\|y|\le L}}
                [g(x-y)-g(x_*-y)].
\end{equation}
\begin{theorem}[Canonical Gibbs property]\label{thm:DLR}
For every limit in Theorem~\ref{thm:limit}, the
limit~\eqref{eq:potential} exists almost surely, uniformly for
$x\in\overline\Lambda$, as the real radius $L\to\infty$.
Conditionally on $(\eta,k)$, the inside points have density on
$\Lambda^k$, interpreted as an unlabeled configuration, proportional
to
\begin{equation}\label{eq:gibbs}
 \exp\left\{-\beta\left[
        -\sum_{i<j}|x_i-x_j|^a+\sum_{i=1}^kW_\eta(x_i)
                         \right]\right\}.
\end{equation}
Changing $x_*$ or the center of the symmetric cutoff does not
change this normalized density. The resulting number-preserving
kernels are consistent for nested intervals on their common full-measure
class of configurations with the stated principal values.
\end{theorem}

The conditioning in this theorem includes the inside count. No
insertion chemical potential is required. The theorem applies to
every periodic accumulation point and does not assert uniqueness
or convergence of the full sequence.

The matching also controls two directly observable quantities. For
$K=C_a/\beta$, define
\[
 V(K)=\bigl(2\sqrt{2K+\sqrt K}+1/2\bigr)^2.
\]
\begin{corollary}[Number fluctuations and Fourier averages]\label{cor:number}
Every process in Theorem~\ref{thm:limit} satisfies
\begin{equation}\label{eq:number-variance}
 \sup_{I\text{ bounded interval}}\Var\xi(I)\le V(K).
\end{equation}
With $A_n=(2n)^{-1}\sum_{y\in\xi\cap[-n,n]}e^{2\pi iy}$,
the variables $A_n$ converge in $L^2$ and
\begin{equation}\label{eq:fourier-bound}
 \lim_{n\to\infty}\E|A_n|^2\ge(1-2\pi^2K)_+^2.
\end{equation}
In particular, the limiting second moment is positive when
$\beta>2\pi^2C_a$.
\end{corollary}

The Fourier conclusion is an averaged statement and requires the
displayed temperature restriction. The all-temperature matching
bound and bounded number variance hold without it.

\subsection{Proof organization}

Section~\ref{sec:finite} proves Theorem~\ref{thm:finite} by the
random-conductance comparison. Section~\ref{sec:matching} constructs
the joint matching for every periodic limit and proves
Corollary~\ref{cor:number}. The remaining sections establish the
Gibbs property. We first isolate the kernel estimates and the
reference-lattice cancellation in Section~\ref{sec:potential}.
Section~\ref{sec:intrinsic} identifies the resulting potential with
the ordinary spatial principal value and constructs measurable
canonical kernels. Section~\ref{sec:dlr} proves their convergence
from the finite-circle conditional laws.

\section{Convexity and the resistance bound}
\label{sec:finite}

We first prove Theorem~\ref{thm:finite}.  The main estimate compares
the inverse random Hessian with the effective resistance of a fixed
long-range network.  Log-concavity controls the physical lengths of
arcs, which in turn controls the expected inverse conductances.  A
deterministic flow then transfers this information to the inverse
Hessian appearing in the Brascamp--Lieb inequality.

\subsection{The periodic kernel and cyclic gaps}

We record the justification of \eqref{eq:images}, including its
normalization.  For $M\ge1$, set
\[
 K_{N,M}(r)=-\sum_{\ell=-M}^{M}
                  \bigl(|r+\ell N|^a-|\ell N|^a\bigr).
\]
Taylor's theorem bounds a paired summand, uniformly on each compact
set of $r$, by a constant times $\ell^{a-2}$.  Since $a-2<-1$,
the sequence $K_{N,M}$ converges locally uniformly to the continuous
even function $K_N$ on the right side of \eqref{eq:images}.  Also,
\[
 K_{N,M}(r+N)-K_{N,M}(r)
       =-|r+(M+1)N|^a+|r-MN|^a\longrightarrow0,
\]
locally uniformly in $r$, because $a-1<0$.  Thus $K_N$ is
$N$-periodic and $K_N(0)=0$.

For a positive integer $k$, write $\omega=2\pi k/N$ and
$L_M=(M+1/2)N$.  Translation of the integration cells gives
\begin{align*}
 \frac1N\int_{-N/2}^{N/2}K_{N,M}(r)e^{-i\omega r}\,dr
 &=-\frac2N\int_0^{L_M}t^a\cos(\omega t)\,dt\\
 &=\frac{2a}{N\omega}
                  \int_0^{L_M}t^{a-1}\sin(\omega t)\,dt.
\end{align*}
The integration-by-parts boundary term is zero, since
$\sin(\omega L_M)=0$.  Dirichlet's test gives convergence of the
last integral at infinity, and exponential damping evaluates it as
\[
 \int_0^\infty t^{a-1}\sin(\omega t)\,dt
       =\Gamma(a)\omega^{-a}\sin(\pi a/2).
\]
Consequently the nonzero Fourier coefficients of $K_N$ are
$c_aN^a|k|^{-1-a}$, with exactly the $c_a$ in \eqref{eq:kernel}.
Evenness gives the negative coefficients.  The continuous periodic
functions $K_N$ and $g_N$ therefore differ by a constant; evaluating
at zero proves \eqref{eq:images}.

The same image series may be differentiated twice on compact
subsets of $(0,N)$: its differentiated tails converge locally
uniformly.  This yields
\begin{equation}\label{eq:curvature}
 g_N''(r)=a(1-a)\sum_{\ell\in\Z}|r+\ell N|^{a-2}>0,
                     \qquad 0<r<N.
\end{equation}
In particular, the curvature is obtained from a convergent image
series.  The kernel itself remains finite and continuous at a
collision.

Let
\[
 \mathcal D_N=\{(Y_1,\ldots,Y_{N-1}):
                     0<Y_1<\cdots<Y_{N-1}<N\}
\]
be the anchored chamber.  Define the cyclic gaps by $D_i=Y_{i+1}-Y_i$, where $Y_0=0$ and
$Y_N=N$.  For $0<u<1$, also write
\[
 \mathcal D_{N,u}=\{Y\in\overline{\mathcal D_N}:
                                      D_i\ge u\text{ for all }i\}
\]
and let $\E_u$ denote expectation for the density proportional to
$e^{-\beta H_N}$ restricted to this set.  It has nonempty interior,
and every pair separation stays away from both $0$ and $N$.
The unrestricted expectation is denoted by $\E$.

\begin{lemma}[Convexity and arc moments]\label{lem:convexity}
The anchored Hamiltonian is strictly convex on $\mathcal D_N$,
and the Palm law is log-concave.  Under either the Palm law or any
of the laws $\E_u$, every cyclic gap has mean one.  If $R$ is the
sum of $k$ consecutive cyclic gaps, $1\le k\le N$, then
\begin{equation}\label{eq:arc-moments}
 \E R=k,\qquad \E R^2\le2k^2,\qquad
 \E R^{2-a}\le2^{(2-a)/2}k^{2-a}.
\end{equation}
The same inequalities hold with $\E_u$ in place of $\E$.
\end{lemma}
\begin{proof}
For a vector $v\in\R^{N-1}$, extended by $v_0=0$,
\begin{equation}\label{eq:hessian-form}
 v^{\mathsf T}\nabla^2H_N(Y)v
   =\sum_{0\le i<j<N}g_N''(Y_j-Y_i)(v_j-v_i)^2.
\end{equation}
Equation~\eqref{eq:curvature} makes this positive unless $v=0$.
The chamber and its truncated versions are convex, so their
Gibbs densities, extended by zero, are log-concave.

Cyclically rotating the gaps is an affine, volume-preserving map
of the gap simplex.  It represents re-rooting at the next particle,
and preserves the pair energy.  It also preserves each condition
$D_i\ge u$.  All gaps therefore have equal means under either
law.  Their sum is $N$, giving mean one and hence $\E R=k$.

An affine marginal of a log-concave measure is log-concave, by the
Pr\'ekopa--Leindler theorem; see, for example,~\cite{BL}.  This
applies to $R$, including when a cyclic arc passes through the
anchored seam.  Its survival function $S(t)=\Prob(R>t)$ is also
log-concave on $[0,\infty)$ and satisfies $S(0)=1$.  Concavity of
$\log S$ implies
\[
 S(x+y)\le S(x)S(y),\qquad x,y\ge0.
\]
Indeed, the increment of a concave function over an interval of
length $y$ decreases as the interval moves to the right; apply
this with initial points $0$ and $x$.  The statement remains true
when a survival probability is zero.  Tonelli's theorem now gives
\[
 \frac12\E R^2
  =\int_0^\infty tS(t)\,dt
  =\int_0^\infty\!\int_0^\infty S(x+y)\,dx\,dy
  \le\left(\int_0^\infty S(t)\,dt\right)^2=k^2.
\]
Since $1<2-a<2$, monotonicity of $L^p$ norms proves the last
inequality in \eqref{eq:arc-moments}.  Degenerate marginals obey
these estimates as well.
\end{proof}

\subsection{The inverse Hessian as a network resistance}

On the vertex set $\{0,\ldots,N-1\}$, give the edge $\{i,j\}$
the random conductance
\[
 c_{ij}(Y)=g_N''(Y_j-Y_i),\qquad i<j.
\]
The full network Laplacian $L_c$ has off-diagonal entries
$-c_{ij}$ and diagonal entries $\sum_{j\ne i}c_{ij}$.
Let $A$ be its restriction obtained by removing row and column
zero.  Equation~\eqref{eq:hessian-form} identifies $A$ exactly
with the Hessian of $H_N$ in $Y_1,\ldots,Y_{N-1}$.

We use the convention that flow energy is summed over unordered
edges.  For an antisymmetric flow $\theta$, this energy is
\[
 \mathcal E_c(\theta)=\sum_{\{v,w\}}
                               \frac{\theta_{vw}^2}{c_{vw}}.
\]
The effective resistance $R_c(0,i)$ is the minimum of this energy
over unit flows from $i$ to $0$.  Reversing the orientation does
not change the minimum.  Thomson's principle~\cite{LP} and the
grounded inverse identity give
\begin{equation}\label{eq:grounded-resistance}
 R_c(0,i)=[A^{-1}]_{ii}.
\end{equation}
For clarity, solve $Av=e_i$ and set $v_0=0$.  The current
$c_{vw}(v_v-v_w)$ is a unit flow from $i$ to $0$, with energy
$v^{\mathsf T}L_cv=v_i$.  Its energy cross term with any
divergence-free flow vanishes by summation by parts.  Every other
unit flow therefore has at least this energy, proving both
statements with the displayed normalization.

Define the cyclic label distance and the deterministic conductances
by
\[
 d_N(i,j)=\min\{|i-j|,N-|i-j|\},\qquad
 w_{ij}=d_N(i,j)^{a-2}.
\]
\begin{lemma}[Expected resistance comparison]\label{lem:network}
There are constants $B_a,D_a<\infty$, depending only on $a$, such
that, uniformly in $N\ge2$ and $0\le i<N$,
\begin{equation}\label{eq:resistance-bound}
 \E R_c(0,i)\le B_aR_w(0,i),\qquad R_w(0,i)\le D_a.
\end{equation}
The first inequality also holds under every $\E_u$.  One may take
\[
 B_a=\frac{2^{(2-a)/2}}{a(1-a)},\qquad
 D_a=\frac{2^{3-a}}{a\eta_a},\qquad
 \eta_a=\min\left\{\frac{3^{2-a}}{24},1-\frac{\sqrt3}{2}\right\}.
\]
\end{lemma}
\begin{proof}
For distinct labels $v,w$, choose the cyclic arc consisting of
$d_N(v,w)$ gaps, and let $R_{vw}$ be its physical length.
One of the terms in \eqref{eq:curvature} has argument of absolute
value $R_{vw}$.  Consequently
\[
 c_{vw}\ge a(1-a)R_{vw}^{a-2},\qquad
 \E\frac1{c_{vw}}
       \le\frac{\E R_{vw}^{2-a}}{a(1-a)}
       \le\frac{B_a}{w_{vw}},
\]
by Lemma~\ref{lem:convexity}.  These inequalities also hold with
$\E_u$.

Choose a deterministic energy-minimizing unit flow $\theta$ in the
$w$ network.  It is an admissible flow in each realization of the
$c$ network, so
\[
 \E R_c(0,i)
  \le\sum_{\{v,w\}}\theta_{vw}^2\E\frac1{c_{vw}}
  \le B_a\sum_{\{v,w\}}\frac{\theta_{vw}^2}{w_{vw}}
  =B_aR_w(0,i).
\]
This proves the comparison using only expected reciprocal
conductances.

It remains to bound the resistance of the deterministic network.
Its Laplacian is circulant.  With the orthonormal Fourier vectors
$N^{-1/2}(e^{2\pi imj/N})_{j=0}^{N-1}$, its nonzero eigenvalues are
\begin{equation}\label{eq:cycle-eigenvalues}
 \lambda_m=\sum_{r=1}^{N-1}d_N(0,r)^{a-2}
                         (1-\cos(2\pi mr/N)),
                         \qquad 1\le m<N.
\end{equation}
For $1\le m\le N/12$, restrict this sum to the integers in
\[
 \frac{N}{6m}\le r\le\frac{N}{3m}.
\]
There are at least $N/(12m)$ such integers: the left endpoint is
at least two, and the interval's length is $N/(6m)$.
On this interval $r\le N/2$,
$1-\cos(2\pi mr/N)\ge1/2$, and
$r^{a-2}\ge(N/(3m))^{a-2}$.  Hence
\[
 \lambda_m\ge\frac{3^{2-a}}{24}(m/N)^{1-a}.
\]
For $N/12\le m\le N/2$, the term $r=1$ instead gives
\[
 \lambda_m\ge1-\frac{\sqrt3}{2}
          \ge\left(1-\frac{\sqrt3}{2}\right)(m/N)^{1-a}.
\]
These two estimates cover every $1\le m\le\lfloor N/2\rfloor$,
including $2\le N<12$, when the first range is empty.
Finally, $\lambda_{N-m}=\lambda_m$.

The pseudoinverse of the full Laplacian, evaluated against
$e_i-e_0$, therefore gives
\begin{align*}
 R_w(0,i)
  &=\frac1N\sum_{m=1}^{N-1}
           \frac{|1-e^{2\pi imi/N}|^2}{\lambda_m}\\
  &\le\frac8{\eta_aN}
          \sum_{m=1}^{\lfloor N/2\rfloor}(N/m)^{1-a}
   \le\frac{2^{3-a}}{a\eta_a}.
\end{align*}
The last step uses
$\sum_{m=1}^{M}m^{a-1}\le\int_0^M t^{a-1}\,dt=M^a/a$.
This proves the assertion for all $N$.
\end{proof}

\subsection{Brascamp--Lieb on the anchored chamber}

The Brascamp--Lieb inequality has a classical role in positional
order for the one-dimensional plasma~\cite{BL1975,BL}.  We use its
convex-domain version: if $K$ is a smooth convex body and $V$ is
smooth with positive-definite Hessian on a neighborhood of $K$,
then the probability measure proportional to $e^{-V}\mathbf1_K$
satisfies
\begin{equation}\label{eq:BL-domain}
 \Var(f)\le\E\bigl\langle(\nabla^2V)^{-1}\nabla f,\nabla f\bigr\rangle
\end{equation}
for every $C^1$ function $f$.  In Euclidean dimension at least two,
this is the case with generalized dimension parameter $+\infty$ in
Kolesnikov--Milman~\cite[Theorem~1.2(1), equivalently
Theorem~3.1(1)]{KM}.  There is no boundary condition on $f$.

For completeness, the same inequality on an interval, which is
the case needed when $N=2$, follows directly by integration by
parts.  Write $\rho=Z^{-1}e^{-V}$ on $[b,c]$, set
$h=f-\E f$, and define
\[
 w(x)=\rho(x)^{-1}\int_b^x h(s)\rho(s)\,ds.
\]
Then $w(b)=w(c)=0$ and $h=w'-V'w$.  Integration by parts gives
\[
 \int h^2\rho
   =\int\bigl((w')^2+V''w^2\bigr)\rho,
       \qquad
 \int h^2\rho=-\int f'w\rho.
\]
Weighted Cauchy--Schwarz, using $V''>0$, proves
$\int h^2\rho\le\int(f')^2/V''\,\rho$, namely
\eqref{eq:BL-domain} in dimension one.

Here the collision singularities and the corners of the chamber
can be handled in a fixed order.  Fix $0<u<1$ first.  The function
$H_N$ is smooth on a neighborhood of $\mathcal D_{N,u}$, and its
Hessian is uniformly positive definite there after shrinking that
neighborhood.  To pass \eqref{eq:BL-domain} to this polytope, one
may use the smooth convex sublevel sets
\[
 K_{u,t}=\left\{Y:
             \sum_{i=0}^{N-1}e^{-t(D_i-u)}\le1\right\},
                   \qquad t>\frac{\log N}{1-u}.
\]
They lie in the interior of $\mathcal D_{N,u}$ and have nonempty
interior, since the equal-gap configuration makes the displayed
sum less than one.  The defining function is strictly convex:
the gradients of the affine gaps span $\R^{N-1}$.  Its only
critical point is its minimizer, so its level set at one is smooth.
As $t\to\infty$, the sets increase to the interior of
$\mathcal D_{N,u}$.  At fixed $u$, both $e^{-\beta H_N}$ and the
inverse Hessian are bounded and continuous on this compact
polytope.  Dominated convergence therefore passes
\eqref{eq:BL-domain} to its restricted Gibbs law.  This is also
the convex-domain approximation described in~\cite[Section~3.1.1]{KM}.

\begin{proof}[Proof of the quadratic bound in Theorem~\ref{thm:finite}]
The coordinate $Y_0$ is identically zero.  For $1\le i<N$, apply
\eqref{eq:BL-domain} on $\mathcal D_{N,u}$ to $V=\beta H_N$ and
$f(Y)=Y_i$.  Equations~\eqref{eq:grounded-resistance} and
\eqref{eq:resistance-bound} give
\[
 \Var_u(Y_i)
    \le\beta^{-1}\E_u[A^{-1}]_{ii}
    =\beta^{-1}\E_uR_c(0,i)
    \le\frac{B_aD_a}{\beta}.
\]
Cyclic symmetry gives $\E_uY_i=i$ for every $u$.  The constants
in the last inequality are independent of $u,N$ and $i$.

Now fix $N$ and let $u\downarrow0$.  The potential is continuous
on the compact closed chamber, and every coordinate is bounded
there.  The restricted normalizing constants, first moments and
second moments thus converge to those of the Palm law by dominated
convergence.  This proves \eqref{eq:finite} with
$C_a=B_aD_a$.
\end{proof}

\subsection{Exponential displacement tails}

The last step uses only the already established variance bound
and the log-concavity of each anchored coordinate marginal.

\begin{lemma}\label{lem:exptail}
If $X$ is a real log-concave random variable with mean zero and
$\E X^2\le\sigma^2$ for $\sigma>0$, then
\[
 \Prob(|X|>t)\le2e^{-\kappa t/\sigma},\qquad
 \E e^{\kappa|X|/(2\sigma)}\le3,
                 \qquad t\ge0,\quad \kappa=\frac{\log3}{4}.
\]
\end{lemma}
\begin{proof}
A degenerate $X$ equals zero and satisfies both conclusions.
Otherwise its survival function $S(t)=\Prob(X\ge t)$ is
log-concave.  Chebyshev's inequality gives
\[
 S(2\sigma)\le\frac14,\qquad S(-2\sigma)\ge\frac34.
\]
For $t\ge2\sigma$, concavity of $\log S$ bounds the later
secant slope by the slope between $-2\sigma$ and $2\sigma$.
If $S(2\sigma)=0$ the desired bound is immediate; otherwise,
\[
 S(t)\le\frac14
       \exp\left\{-\frac{\log3}{4\sigma}(t-2\sigma)\right\}
       \le e^{-\kappa t/\sigma}.
\]
Apply the same argument to $-X$ and add the two tails.  For
$0\le t<2\sigma$, the bound follows from
$2e^{-\kappa t/\sigma}\ge2/\sqrt3>1$.
Finally, with $\theta=\kappa/(2\sigma)$, Tonelli's theorem gives
\[
 \E e^{\theta|X|}
   =1+\theta\int_0^\infty e^{\theta t}\Prob(|X|>t)\,dt
   \le1+\frac{2\theta}{\kappa/\sigma-\theta}=3.
\]
\end{proof}
By Lemma~\ref{lem:convexity}, each $Y_i-i$ is a log-concave
marginal.  Lemma~\ref{lem:exptail}, with
$\sigma=(C_a/\beta)^{1/2}$, completes the proof of
Theorem~\ref{thm:finite}.

\section{Stationary matching and number fluctuations}\label{sec:matching}

The finite estimate is rooted at a particle, whereas a thermodynamic
limit is stationary in space. We first retain both the displacement
field and its lattice phase while removing the root. This also
prevents particles with remote labels from being lost in a local
limit.

\subsection{Joint extraction and simplicity}

\begin{proof}[Proof of Theorem~\ref{thm:limit}]
For the Palm configuration extend $Y_i$ by $Y_{i+N}=Y_i+N$ and
put $p_i=Y_i-i$, so $p_{i+N}=p_i$. Independently choose a uniform
$J\in\{0,\ldots,N-1\}$ and set
\[
 q_j^{(N)}=p_{j+J},\qquad j\in\Z.
\]
This sequence is stationary under integer shifts. With an independent
uniform $U_N\in[0,1)$, define
\[
 y_j^{(N)}=j+U_N+q_j^{(N)}
          =Y_{j+J}+U_N-J.
\]
As a point set this is the periodic Palm configuration translated by
$U_N-J$. The latter variable is uniform modulo $N$ and independent
of the Palm configuration before the label shift. Undoing the Palm
rooting therefore shows that $\sum_j\delta_{y_j^{(N)}}$ has exactly
law $\mu_N$.

Theorem~\ref{thm:finite} gives, for every $N,j$,
\begin{equation}\label{eq:uniform-q}
 \E|q_j^{(N)}|^2\le K,\qquad
 \Prob(|q_j^{(N)}|>t)\le2e^{-\kappa t/\sigma},\qquad
 \E e^{\kappa|q_j^{(N)}|/(2\sigma)}\le3.
\end{equation}
These inequalities follow by averaging over $J$; no log-concavity
of the resulting mixture is needed. Product tightness gives tightness
of $(U_N,q^{(N)})$ in $[0,1]\times\R^{\Z}$. Its weak limits have
uniform $U$, independent of the integer-stationary field $q$.
The moment inequalities pass by lower semicontinuity. For the tail
bound apply the Portmanteau inequality to the open set
$\{|q_j|>t\}$.

The point-process laws themselves are tight: for every compact
interval $B$, stationarity at intensity one gives
$\E\xi_N(B)=|B|$, and Markov's inequality on an increasing sequence
of compact intervals is the usual tightness criterion for locally
finite counting measures. Start with any convergent point-process
subsequence and take a further joint subsequence of
$(U_N,q^{(N)},\xi_N)$. A Skorokhod representation on this Polish
product space makes all these variables converge almost surely.

We identify the limiting point process. For $M>R+2$, a label
$|j|>M$ can enter $[-R,R]$ only if
$|q_j^{(N)}|\ge |j|-R-1$. Hence
\begin{equation}\label{eq:remote-labels}
 \E\sum_{|j|>M}{\bf1}_{\{y_j^{(N)}\in[-R,R]\}}
 \le K\sum_{|j|>M}(|j|-R-1)^{-2}.
\end{equation}
The right side tends to zero uniformly in $N$, and the same bound
holds for $y_j=j+U+q_j$. In particular the limiting sum
$\sum_j\delta_{y_j}$ is locally finite almost surely. For each
compactly supported continuous test function, first restrict the
sum to $|j|\le M$, where coordinate convergence applies, and then
use~\eqref{eq:remote-labels} on the omitted terms. The joint limit
therefore obeys $\xi=\sum_j\delta_{y_j}$ almost surely. This
identification preserves the originally selected marginal limit.

The finite gaps $1+q_{j+1}^{(N)}-q_j^{(N)}$ are positive, so the
limiting gaps are nonnegative. Their marginal law is the common
cyclic gap law of the Palm ensemble. This law is log-concave by
Lemma~\ref{lem:convexity}, has mean one and second moment at most
two. The second-moment bound gives uniform integrability of the
gaps, so every limiting gap has mean one. Weak limits of log-concave
probability measures are log-concave. In one dimension such a law
is either a point mass or has a log-concave density on its affine
support. In the latter case it has no atom at zero; in the former
case its mean forces the mass to lie at one. Consequently every
limiting gap is strictly positive almost surely. A countable
intersection establishes strict ordering for all labels and hence
simplicity.

For a nonnegative compactly supported function $f$, Tonelli's theorem
and stationarity of the field give
\begin{align*}
 \E\sum_{j\in\Z}f(j+U+q_j)
 &=\sum_{j\in\Z}\int_0^1\E f(j+u+q_0)\,du\\
 &=\E\int_\R f(t+q_0)\,dt
 =\int_\R f(t)\,dt.
\end{align*}
Thus the limiting intensity is exactly one.

Finally translate the matching by $t\in\R$. Write
$U+t=m+U'$ with $m\in\Z$ and $U'\in[0,1)$. Reindexing by
$j'=j+m$ sends the translated matching to the representation with
phase $U'$ and field $q'_{j'}=q_{j'-m}$. Conditional on $U$, the
field is independent and integer-stationary, so the law of $q'$ is
the original law and is independent of $U'$. The latter is again
uniform. This proves joint spatial stationarity.
\end{proof}

\subsection{Crossing currents}

Let $\ell_j=j+U$ and $y_j=\ell_j+q_j$. At a cut $t$ define the
signed crossing count
\begin{equation}\label{eq:current}
 C(t)=\#\{j:\ell_j\le t<y_j\}
       -\#\{j:y_j\le t<\ell_j\}.
\end{equation}
For a deterministic $t$, both sets are finite almost surely. Indeed,
joint stationarity and Tonelli give
\begin{equation}\label{eq:crossing-expectations}
 \E\#\{j:[\ell_j,y_j]\text{ crosses }t\}=\E|q_0|,
 \qquad
 \E\sum_{j:[\ell_j,y_j]\text{ crosses }t}|q_j|=\E q_0^2.
\end{equation}
For example, sum over $j$, integrate the independent phase $U$,
and use the common marginal law of $q_j$. The resulting integral
over all possible lattice origins has length $|q_0|$ for the first
identity and integral $|q_0|^2$ for the second.

Order preservation makes the crossing geometry particularly useful.
A rightward crossing from $\ell_i$ and a leftward crossing from
$\ell_j$ would force $i<j$ and $y_i>y_j$, which is impossible.
Thus all crossings have the same direction; see Figure~\ref{fig:crossings}.
Their lattice endpoints
are consecutive: for a rightward crossing set, if $i<j$ both cross,
then for every $i<h<j$, $\ell_h\le t$ and $y_h\ge y_i>t$.
The leftward case is identical. This is the geometric reason that
a square-integrable displacement controls the square of the crossing
count.

\begin{figure}[t]
\centering
\begin{tikzpicture}[x=1.35cm,y=1.05cm,>=stealth]
 \fill[gray!12] (2,-.35) rectangle (4.5,1.45);
 \draw[gray] (-.3,1) -- (6,1);
 \draw[gray] (-.3,0) -- (6,0);
 \draw[dashed,thick] (3.5,-.48) -- (3.5,1.55)
    node[above] {$t$};
 \foreach \j/\v in {0/.4,1/1.6,2/3.8,3/4.4,4/4.9,5/5.7}{
   \draw[->,gray] (\j,1) -- (\v,0);
   \fill (\j,1) circle (1.6pt);
   \fill (\v,0) circle (1.6pt);
 }
 \draw[->,very thick] (2,1) -- (3.8,0);
 \draw[->,very thick] (3,1) -- (4.4,0);
 \node[left] at (-.35,1) {$\ell_j$};
 \node[left] at (-.35,0) {$y_j$};
 \node[above] at (2,1.02) {$\ell_i$};
 \node[above] at (3,1.02) {$\ell_{i+1}$};
 \node[below,align=center] at (3.15,-.52)
  {crossed lattice endpoints are consecutive};
\end{tikzpicture}
\caption{An ordered matching across a cut. The two thick segments
 contribute with the same sign to $C(t)$; opposite-direction
 crossings are excluded by the ordering of both rows. Crossings
 are also precisely the membership discrepancies between lattice
 and particle cutoffs. For a cutoff $t=\pm L$, logarithmic
 displacements confine their endpoints to a strip of width
 $O(\log L)$ by~\eqref{eq:log-envelope}, as used in the spatial
 principal-value proof.}
\label{fig:crossings}
\end{figure}
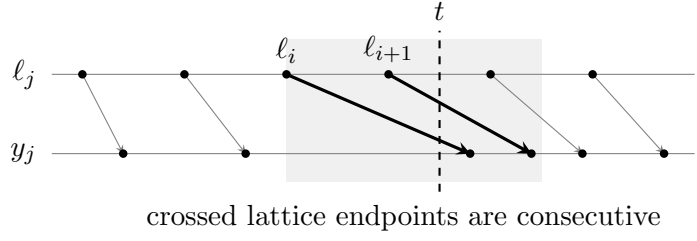

\begin{lemma}\label{lem:current}
If $\E q_0^2\le K$, then for every deterministic cut $t$,
\begin{equation}\label{eq:current-bound}
 \E C(t)^2\le 2K+\sqrt K.
\end{equation}
\end{lemma}
\begin{proof}
Write $B=|C(t)|$. By the preceding ordering observation, $B$ is
the total number of crossings. The distances of the $B$ crossed
lattice endpoints from $t$ are bounded below, after ordering, by
$0,1,\ldots,B-1$. Each matching displacement is at least the
distance of its endpoint from $t$. Therefore
\[
 \frac{B(B-1)}2
 \le\sum_{j:[\ell_j,y_j]\text{ crosses }t}|q_j|.
\]
Take expectations, use~\eqref{eq:crossing-expectations}, and apply
$\E|q_0|\le\sqrt K$.
\end{proof}

\begin{proof}[Proof of Corollary~\ref{cor:number}]
Let $\mathcal L=\sum_j\delta_{\ell_j}$ be the shifted unit lattice.
For $I=(u,v]$, the contribution of each segment to the change of
membership gives
\begin{equation}\label{eq:interval-current}
 \xi(I)-\mathcal L(I)=C(u)-C(v).
\end{equation}
Both counts have mean $|I|$. If the fractional part of $|I|$ is
$r$, the lattice count takes its two neighboring integer values
with probabilities $1-r$ and $r$, so
$\Var\mathcal L(I)=r(1-r)\le1/4$. Minkowski's inequality and
Lemma~\ref{lem:current} imply
\[
 \|\xi(I)-|I|\|_2
 \le\|C(u)\|_2+\|C(v)\|_2
      +\|\mathcal L(I)-|I|\|_2
 \le2\sqrt{2K+\sqrt K}+\frac12.
\]
Stationarity and intensity one exclude particles at any deterministic
endpoint, so this proves~\eqref{eq:number-variance} for all endpoint
conventions.

For the Fourier average, replace the particle cutoff $[-n,n]$ by
the corresponding lattice cutoff. The number of altered terms is
at most $|C(-n)|+|C(n)|$. Their complex absolute values equal one;
hence the normalized error has $L^2$ norm $O(n^{-1})$. Almost
surely $0<U<1$, so the lattice cutoff consists of labels
$-n,\ldots,n-1$. Since $e^{2\pi iy_j}=e^{2\pi iU}e^{2\pi iq_j}$,
the mean ergodic theorem for the stationary sequence gives
\begin{equation}\label{eq:ergodic-amplitude}
 A_n\longrightarrow
 e^{2\pi iU}\E[e^{2\pi iq_0}\mid\mathcal I_{\Z}]
 \quad\text{in }L^2,
\end{equation}
where $\mathcal I_{\Z}$ is its shift-invariant sigma-field.
Conditional Jensen and $\cos z\ge1-z^2/2$ yield
\begin{align*}
 \lim_n\E|A_n|^2
 &=\E\left|\E[e^{2\pi iq_0}\mid\mathcal I_{\Z}]\right|^2\\
 &\ge\left|\E e^{2\pi iq_0}\right|^2
 \ge(1-2\pi^2K)_+^2.
\end{align*}
This establishes the stated averaged bound without an ergodicity
assumption.
\end{proof}

\section{Uniform relative-potential estimates}\label{sec:potential}

The interaction $g(r)=-|r|^a$ grows at infinity. Individual
particle--exterior energies therefore diverge, but differences of
energies at fixed inside particle number have a finite symmetric
limit. This section gives an estimate that is uniform in the circle
size and in the inside evaluation points. The uniformity allows
the exterior configuration itself to vary in the thermodynamic
limit. Auxiliary constants such as $C(a)$ may change from line to line;
the constant $C_a$ in the main theorems keeps its fixed meaning.

\subsection{Kernel bounds and the reference lattice}

For $r,r'\in\R$ let $d_N(r,r')$ denote their distance on the
circle $\mathbb T_N$.
\begin{lemma}\label{lem:kernel-bounds}
There is a finite constant $b_a$, independent of $N$, such that
\begin{align}
 |g_N(r)-g_N(r')|&\le b_a d_N(r,r')^a
                         \le b_a|r-r'|^a,\label{eq:holder}\\
 0<g_N''(r)&\le b_a d_N(r,0)^{a-2},
                         \qquad r\notin N\Z.\label{eq:curvature-upper}
\end{align}
For $L>0$ and $N\ge4L$,
\begin{equation}\label{eq:local-kernel}
 \sup_{|r|\le L}|g_N(r)-g_N(0)-g(r)|
                    \le b_aL^2N^{a-2}.
\end{equation}
The function $g$ itself satisfies the first inequality with
Euclidean distance and constant one.
\end{lemma}
\begin{proof}
The nearest term in the curvature image series has size
$a(1-a)d_N(r,0)^{a-2}$. The remaining terms have total at most
$C(a)N^{a-2}$, which is bounded by a constant times the nearest
term. This proves~\eqref{eq:curvature-upper}.
Evenness and periodicity imply $g_N'(N/2)=0$.
Integrating the curvature estimate from $r$ to $N/2$ gives
\[
 |g_N'(r)|\le C(a) r^{a-1},\qquad 0<r\le N/2.
\]
The reflected bound holds on the other half-circle. Integrating
along a shortest arc, and splitting it at a possible cusp, proves
the uniform H\"older estimate. The singularity $r^{a-1}$ is
integrable since $a>0$.

The function $h_N(r)=g_N(r)+|r|^a$ is even and twice continuously
differentiable on $(-N,N)$, including zero. By the image formula,
\[
 h_N''(r)=a(1-a)\sum_{m\ne0}|r+mN|^{a-2}
            \le C(a)N^{a-2},\qquad |r|\le N/4.
\]
Since $h_N'(0)=0$, two integrations prove~\eqref{eq:local-kernel}.
Finally $||r|^a-|r'|^a|\le|r-r'|^a$ follows from subadditivity
of $t\mapsto t^a$ on $[0,\infty)$.
\end{proof}

\begin{lemma}[Reference-lattice identity]\label{lem:lattice}
For any block $I_N$ of $N$ consecutive integers and any $U,x,x'\in\R$,
\begin{equation}\label{eq:lattice-cancel}
 \sum_{j\in I_N}
 [g_N(x-j-U)-g_N(x'-j-U)]
            =g_1(x-U)-g_1(x'-U).
\end{equation}
Moreover,
\begin{equation}\label{eq:line-lattice}
 g_1(x-U)-g_1(x'-U)
  =\lim_{m\to\infty}\sum_{j=-m}^m
       [g(x-j-U)-g(x'-j-U)],
\end{equation}
locally uniformly in $(x,x',U)$.
\end{lemma}
\begin{proof}
The absolutely convergent Fourier series~\eqref{eq:kernel} can be
summed over the $N$ integer shifts. Only frequencies divisible by
$N$ remain. Their coefficient is
$Nc_aN^a|Nm|^{-1-a}=c_a|m|^{-1-a}$, giving the period-one
kernel exactly. This proves~\eqref{eq:lattice-cancel}.

To prove~\eqref{eq:line-lattice}, pair the terms with labels $j$
and $-j$. For bounded $t,U$, the derivative in $t$ of
\[
 g(t-j-U)+g(t+j-U)
       =-(j+U-t)^a-(j-U+t)^a
\]
is $O(j^{a-2})$, uniformly on those bounded sets once $j$ is large.
The paired differences at $t=x,x'$ are therefore summable. The
image formula~\eqref{eq:images} for $N=1$, applied at $x-U$ and
$x'-U$ and subtracted, identifies their sum with the right side
of~\eqref{eq:lattice-cancel}.
\end{proof}

The surviving period-one term is generally nonconstant. Retaining
it fixes the normalization of the limiting relative potential and
rules out an additional affine term.

\subsection{The displacement correction}

For $x,x'\in K_R=[-R,R]$, define
\[
 \delta_N(x,x';y)=g_N(x-y)-g_N(x'-y),\qquad
 \delta(x,x';y)=g(x-y)-g(x'-y).
\]
Let $\ell_j=j+U$, $y_j=\ell_j+q_j$, and put
\[
 D_{N,j}(x,x')=\delta_N(x,x';y_j)-\delta_N(x,x';\ell_j),
 \qquad
 D_j(x,x')=\delta(x,x';y_j)-\delta(x,x';\ell_j).
\]
In the periodic case, $q$ is $N$-periodic and $I_N$ denotes a
centered block of $N$ labels, chosen so that
$|j|\le(N+1)/2$ for $j\in I_N$. We write $\|\cdot\|_R$ for
the uniform norm on $K_R^2$.

\begin{lemma}[Uniform correction tail]\label{lem:tail}
Assume $U\in[0,1]$ and $\sup_j\E q_j^2\le K$. For
$M\ge16(R+2)$,
\begin{equation}\label{eq:tail-line}
 \E\sum_{|j|>M}\|D_j\|_R
       \le C_{a,R}\bigl(\sqrt K\,M^{a-1}+K/M\bigr).
\end{equation}
In the periodic case the same bound holds for
$\E\sum_{j\in I_N,\,|j|>M}\|D_{N,j}\|_R$, uniformly in
$N\ge2M$.
\end{lemma}
\begin{proof}
Split according to $|q_j|\le|j|/4$. On this event every point
of the segment from $\ell_j$ to $y_j$ has Euclidean distance
at least $c|j|$ from $K_R$. For a centered periodic label its
circle distance is also at least $c|j|$: even at the end of
$I_N$, the route through the opposite side of the circle has
length at least a constant multiple of $|j|$.
Two applications of the fundamental theorem of calculus give
\[
 D_{N,j}(x,x')
  =-\int_{\ell_j}^{y_j}\int_{x'}^x g_N''(v-z)\,dv\,dz.
\]
The sign is immaterial for the estimate. By
Lemma~\ref{lem:kernel-bounds},
\[
 \|D_{N,j}\|_R\le C_{a,R}|q_j||j|^{a-2}.
\]
The same calculation uses $g''(r)=a(1-a)|r|^{a-2}$ on the
line. Taking expectations and summing uses
$\E|q_j|\le\sqrt K$ and
$\sum_{j>M}j^{a-2}\le C(a)M^{a-1}$.

On the complementary event, the H\"older bound gives the
deterministic estimate $\|D_{N,j}\|_R\le C_{a,R}$, including
when the displaced point enters $K_R$. Chebyshev gives
$\Prob(|q_j|>|j|/4)\le16K/j^2$.
The sum of these contributions is at most $C_{a,R}K/M$.
\end{proof}

\begin{proposition}\label{prop:full-potential}
Under the full-line assumptions of Lemma~\ref{lem:tail}, the series
\begin{equation}\label{eq:full-line-potential}
 \Phi(x,x')=g_1(x-U)-g_1(x'-U)+\sum_{j\in\Z}D_j(x,x')
\end{equation}
converges almost surely, locally uniformly, and its correction
series converges absolutely in expected uniform norm. It equals
the symmetric label limit
\begin{equation}\label{eq:label-pv}
 \Phi(x,x')=\lim_{m\to\infty}\sum_{j=-m}^m\delta(x,x';y_j).
\end{equation}
For the periodic full potential
$\Phi_N=\sum_{j\in I_N}\delta_N(\cdot,\cdot;y_j)$,
\begin{equation}\label{eq:full-periodic-potential}
 \Phi_N(x,x')=g_1(x-U)-g_1(x'-U)
                    +\sum_{j\in I_N}D_{N,j}(x,x').
\end{equation}
If a coupling has $(U_N,q^{(N)})\to(U,q)$ almost surely in
the product topology and has the uniform second-moment bound
$K$, then, for every fixed $R$,
\begin{equation}\label{eq:full-field-convergence}
 \E\|\Phi_N-\Phi\|_R\longrightarrow0.
\end{equation}
\end{proposition}
\begin{proof}
Lemma~\ref{lem:tail} and Tonelli imply almost-sure summability
of the uniform norms in the correction series outside a finite
label set. The finite terms are continuous, including at
collisions. This proves its local uniform convergence; the same
bound gives convergence in expected uniform norm.
Lemma~\ref{lem:lattice} then proves~\eqref{eq:label-pv} and
the exact periodic decomposition~\eqref{eq:full-periodic-potential}.

For convergence, fix a finite label cutoff $M$. The periodic
relative kernels converge locally uniformly to the line kernel
by~\eqref{eq:local-kernel}. Together with coordinate convergence,
this proves uniform convergence of every term with $|j|\le M$
on the given coupling. The H\"older bound bounds their uniform
norms by a deterministic constant depending on $a,R$, so bounded
convergence also gives convergence in expected norm. The
period-one term converges uniformly by continuity and periodicity
of $g_1$. Finally discard the two correction tails using
Lemma~\ref{lem:tail}, first sending $N\to\infty$ and then
$M\to\infty$. This proves~\eqref{eq:full-field-convergence}.
\end{proof}

\section{Spatial principal values and canonical kernels}\label{sec:intrinsic}

The relative potential just constructed uses a lattice matching
for its proof. We now identify it directly from the particle
configuration. Exponential displacement tails are used only for
the unrestricted spatial-cutoff conclusion; the uniform comparison
in Proposition~\ref{prop:full-potential} needed only second moments.

\begin{lemma}[All spatial cutoff radii]\label{lem:spatial-pv}
Let $(y_j)$ be the representation in Theorem~\ref{thm:limit}.
On a single probability-one event, for all finite $R,Z$,
\begin{equation}\label{eq:spatial-pv}
 \sup_{\substack{x,x'\in[-R,R]\\|z|\le Z}}
 \left|\sum_{|y_j-z|\le L}\delta(x,x';y_j)-\Phi(x,x')\right|
     \longrightarrow0
 \quad\text{as real }L\to\infty.
\end{equation}
The limit is unchanged by finite shifts of the increasing
enumeration or by endpoint conventions. Removing a finite set of
particles subtracts exactly its finite relative-potential sum.
\end{lemma}
\begin{proof}
The uniform exponential tail and Borel--Cantelli give a finite
random $B$ and a deterministic $b$ such that
\begin{equation}\label{eq:log-envelope}
 |q_j|\le B+b\log(2+|j|),\qquad j\in\Z.
\end{equation}
Indeed choose $b>\sigma/\kappa$ so that
$\sum_j\Prob(|q_j|>b\log(2+|j|))<\infty$, then absorb the
finitely many exceptions into $B$. In particular $y_j/j\to1$
at either end of the line.

For $|y|\ge2R+1$, direct differentiation gives
\begin{equation}\label{eq:far-delta}
 \|\delta(\cdot,\cdot;y)\|_R\le C_{a,R}|y|^{a-1},\qquad
 \|\partial_y\delta(\cdot,\cdot;y)\|_R
                                  \le C_{a,R}|y|^{a-2}.
\end{equation}
The logarithmic envelope and the second estimate imply
\[
 \|D_j\|_R\le C_{a,R}|q_j||j|^{a-2}
       \le C_{a,R,\omega}\log(2+|j|)|j|^{a-2}
\]
for all sufficiently large $|j|$. The paired reference-lattice
tail is $O_{a,R}(m^{a-1})$ by Lemma~\ref{lem:lattice}.
Consequently the tail in the symmetric label sum is
$O_{a,R,\omega}(m^{a-1}\log m)$.

As illustrated in Figure~\ref{fig:crossings}, compare the membership conditions $|j|\le L$ and
$|y_j-z|\le L$ for $|z|\le Z$. For sufficiently large $L$,
the envelope excludes every $|j|>2L$ from either spatial window:
$|y_j-z|\ge|j|/2>L$. Among the remaining labels set
\[
 b_L=1+Z+\max_{|j|\le2L}|q_j|=O_\omega(\log L).
\]
Every membership discrepancy must satisfy
$\bigl||j|-L\bigr|\le b_L$. There are at most $4b_L+4$
such labels. Their particle positions have $|y_j|\ge L/2$
once $L$ is sufficiently large, uniformly for $|z|\le Z$.
The first estimate in~\eqref{eq:far-delta} bounds the difference
of the two partial sums by
\[
 C_{a,R}(4b_L+4)L^{a-1}
          =O_{a,R,Z,\omega}(L^{a-1}\log L)\longrightarrow0.
\]
This deterministic comparison holds simultaneously for every real
$L$ above a random threshold, rather than only along a selected
sequence. Together with the label-tail bound it proves
\eqref{eq:spatial-pv}. A countable exhaustion in $R,Z$ supplies
the asserted common probability-one event.

Changing the enumeration origin shifts a symmetric label cutoff
by only finitely many labels at its two ends. The corresponding
terms tend to zero uniformly by~\eqref{eq:far-delta}. Endpoint
changes fall within the same boundary strips. Finite deletion
commutes with every sufficiently large spatial sum and therefore
with its locally uniform limit.
\end{proof}

Deleting the particles in $\Lambda$ from Lemma~\ref{lem:spatial-pv}
gives~\eqref{eq:potential}. In particular the potential is an
intrinsic function of the exterior. Its cutoff-center independence
also gives translation covariance. The finite relative-potential
identities pass to the limit; for instance,
\begin{equation}\label{eq:potential-cocycle}
 W_\eta(x;x_*)=W_\eta(x;x_*')+W_\eta(x_*';x_*).
\end{equation}
Changing the reference point adds a constant independent of $x$.
For a fixed inside count this does not affect a normalized Gibbs
density.

The prescription agrees with subtraction of the symmetric uniform
background. Indeed, for $|x|<L$,
\[
 \int_{-L}^L g(x-y)\,dy
  =-\frac{(L+x)^{a+1}+(L-x)^{a+1}}{a+1}.
\]
Its derivative in $x$ is
$-(L+x)^a+(L-x)^a=O_{a,R}(L^{a-1})$ uniformly for
$|x|\le R$. Therefore
\begin{equation}\label{eq:background}
 \sup_{x,x'\in K_R}
 \left|\int_{-L}^L[g(x-y)-g(x'-y)]\,dy\right|
        \longrightarrow0.
\end{equation}
The same holds for a fixed translated center.

\subsection{Measurability and consistency}

For completeness, define the kernel on the full counting-measure
space, not only on a particular choice of matching. Fix $\Lambda$
and $x_*$. Each finite physical sum
\[
 (\eta,x)\longmapsto
 \sum_{y\in\eta\cap[-m,m]}[g(x-y)-g(x_*-y)]
\]
is jointly Borel and continuous in $x$ on $\overline\Lambda$.
As a map into the separable space $C(\overline\Lambda)$ it is
Borel: evaluation on a countable dense set determines its Borel
sigma-field. The set on which these maps form a uniformly Cauchy
sequence is Borel, again by taking suprema on a countable dense
set. On this set use their uniform limit $W_\eta$; on the
complement put $W_\eta=0$. This gives an exterior-measurable
$C(\overline\Lambda)$-valued version which agrees almost surely
with~\eqref{eq:potential} under every limit under consideration.

For $k\ge1$ define
\begin{align}
 E_{\Lambda,k}(x;\eta)
   &=\sum_{i<j}g(x_i-x_j)+\sum_{i=1}^kW_\eta(x_i),
                         \label{eq:canonical-energy}\\
 \gamma_{\Lambda,k}(dx\mid\eta)
   &=\frac{e^{-\beta E_{\Lambda,k}(x;\eta)}}
      {Z_{\Lambda,k}(\eta)}\,dx_1\cdots dx_k,
                         \label{eq:canonical-density}
\end{align}
on $\Lambda^k$. The energy is bounded and continuous on
$\overline\Lambda^k$, so $0<Z_{\Lambda,k}(\eta)<\infty$.
It is symmetric; pushing forward the labeled tuple to its counting
measure gives the unlabeled law. Equivalently one may divide
Lebesgue measure by $k!$, which cancels from the normalizer.
Collision diagonals and endpoints have zero probability.
For $k=0$, use the empty configuration.

For a configuration $\omega$, let $\gamma_\Lambda$ keep
$\omega|_{\Lambda^c}$ and $k=\omega(\Lambda)$ fixed and resample
inside by~\eqref{eq:canonical-density}. Parameter-dependent Borel
integration shows that this is a Borel probability kernel. For
every bounded Borel observable $F$, $\gamma_\Lambda F$ is
measurable with respect to
\begin{equation}\label{eq:canonical-sigma}
 \mathcal C_\Lambda
   =\sigma\{\omega|_{\Lambda^c},\omega(\Lambda)\}.
\end{equation}
It preserves these data pointwise. Consequently, for every bounded
$\mathcal C_\Lambda$-measurable $H$,
\begin{equation}\label{eq:properness}
 \gamma_\Lambda(FH)=H\gamma_\Lambda F.
\end{equation}

On configurations satisfying the spatial principal-value conclusion,
finite changes preserve convergence. If $\Lambda\subset\Delta$,
the finite-deletion identity gives
\[
 W_{\omega|_{\Lambda^c}}(x;x_*)
 =W_{\omega|_{\Delta^c}}(x;x_*)
    +\sum_{y\in\omega\cap(\Delta\setminus\Lambda)}
            [g(x-y)-g(x_*-y)].
\]
Condition the finite-dimensional density in $\Delta$ on its
particles in $\Delta\setminus\Lambda$. Its remaining interaction
with the inside tuple is exactly the energy for $\Lambda$, up to
terms independent of that tuple. This proves
$\gamma_\Delta\gamma_\Lambda=\gamma_\Delta$.
The reverse composition also equals $\gamma_\Delta$, since
$\gamma_\Lambda$ preserves the exterior of $\Delta$ and the
total count in $\Delta$, the two data on which
$\gamma_\Delta$ depends. Thus the kernels are consistent on the
invariant class used by the Gibbs limits. The arbitrary Borel
completion outside this class is irrelevant to their conditional
laws.

\section{Passage from the circle to the line}\label{sec:dlr}

We now prove Theorem~\ref{thm:DLR}. The convergence required here
is convergence of potentials along the random exterior
configurations selected by the Gibbs laws. The expected uniform
tail bound supplies exactly this form of convergence. The final
conditional-law argument uses fixed local observables on the
circle before extending to the entire line.

\subsection{Deleting the actual inside particles}

Fix a weakly convergent subsequence $\mu_N\Rightarrow\mu$ and
the joint coupling constructed in the proof of
Theorem~\ref{thm:limit}. Thus the phases, every displacement
coordinate, and the point configurations converge almost surely.
Fix $\Lambda=(b,c)$ and choose $R$ with
$\overline\Lambda\subset[-R,R]$. Since the limiting process is
stationary of intensity one, it has no particle at either fixed
endpoint almost surely. Vague convergence therefore gives
\begin{equation}\label{eq:count-stability}
 k_N=\xi_N(\Lambda)=\xi(\Lambda)=k
 \quad\text{eventually almost surely},
 \qquad
 \eta_N=\xi_N|_{\Lambda^c}\longrightarrow
 \eta=\xi|_{\Lambda^c}
\end{equation}
vaguely. These facts follow, for example, by placing small
boundary neighborhoods containing no limiting particles and
testing the compact interval between them. Counts are integer
valued and the limit is locally finite.

For $N>2R$, choose the deterministic fundamental cell
$C_N=[-N/2,N/2)$. A centered label $j\in I_N$ has a unique
representative $\widehat y_j^{(N)}$ in $C_N$. Define
\[
 B_N=\{j\in I_N:\widehat y_j^{(N)}\in\Lambda\},
 \qquad
 B=\{j\in\Z:y_j\in\Lambda\}.
\]
The set $B_N$ includes a particle whenever any of its periodic
representatives lies in $\Lambda$, regardless of the position of
the unwrapped coordinate $y_j^{(N)}$.
If $|j|>M\ge16(R+2)$ and $j\in B_N$, then
\[
 |q_j^{(N)}|\ge d_N(j+U_N,\Lambda)\ge c_0|j|
\]
for an absolute $c_0>0$. The first inequality uses the displacement
of a lattice point on the circle; the second uses the centered
label block. It follows that
\begin{equation}\label{eq:inside-remote}
 \E\#\{j\in B_N:|j|>M\}\le C_RK/M.
\end{equation}
The identical bound holds for $B$ on the line. Thus with
probability tending to one as $M\to\infty$, uniformly in large
$N$, all the particles to be deleted have labels in $[-M,M]$.
For those finitely many labels, coordinate convergence and
boundary avoidance stabilize membership in $\Lambda$. Their
unwrapped coordinates are bounded along the coupling, so for
large $N$ they coincide with their representatives in $C_N$.

The full relative potential in the circle is $\Phi_N$ from
Proposition~\ref{prop:full-potential}. The exterior relative
potentials, as functions of two inside evaluation points, are
\begin{align}
 \Phi_N^{\mathrm{out}}(x,x')
   &=\Phi_N(x,x')-
       \sum_{j\in B_N}\delta_N(x,x';y_j^{(N)}),
                    \label{eq:finite-deletion}\\
 \Phi^{\mathrm{out}}(x,x')
   &=\Phi(x,x')-
       \sum_{j\in B}\delta(x,x';y_j).
                    \label{eq:limit-deletion}
\end{align}
The periodic arguments in~\eqref{eq:finite-deletion} make the
choice of representatives immaterial. For fixed $M$ all the
terms retained in the deletion sums converge uniformly, by
Lemma~\ref{lem:kernel-bounds}. Equation~\eqref{eq:inside-remote}
controls the event of any omitted inside label, while
Proposition~\ref{prop:full-potential} controls the full fields.
Consequently
\begin{equation}\label{eq:outside-convergence}
 \|\Phi_N^{\mathrm{out}}-\Phi^{\mathrm{out}}\|_R
                  \longrightarrow0
                  \quad\text{in probability}.
\end{equation}
This argument also accounts for particles entering the window
through periodic wraparound. No moment estimate conditioned on
a prescribed exterior has been used.

Fixing the second argument at $x_*$ identifies the limiting
function with $W_\eta$. The finite function is precisely
\begin{equation}\label{eq:finite-exterior-potential}
 W_{N,\eta_N}(x)=
   \sum_{z\in\xi_N\cap(C_N\setminus\Lambda)}
          [g_N(x-z)-g_N(x_*-z)].
\end{equation}
Thus~\eqref{eq:outside-convergence} is convergence of the true
circle exterior field to the intrinsic full-line field, uniformly
on $\overline\Lambda$, along the joint random coupling.

\subsection{Stability of the normalized conditional laws}

On the sector $k_N=k$, subtract the constant $g_N(0)$ from each
internal pair and define
\begin{equation}\label{eq:finite-energy}
 E_{N,k}(x;\eta_N)
    =\sum_{i<j}[g_N(x_i-x_j)-g_N(0)]
                   +\sum_{i=1}^kW_{N,\eta_N}(x_i).
\end{equation}
This subtraction leaves the circle conditional density unchanged.
For $k\le m$, Lemma~\ref{lem:kernel-bounds} gives
\begin{align}
 \sup_{x\in\overline\Lambda^k}
 |E_{N,k}(x;\eta_N)-E_{\Lambda,k}(x;\eta)|
 &\le \frac{m(m-1)}2
   \sup_{|r|\le2R}|g_N(r)-g_N(0)-g(r)|\notag\\
 &\quad+m\|\Phi_N^{\mathrm{out}}-
                       \Phi^{\mathrm{out}}\|_R
 \longrightarrow0
 \quad\text{in probability}.
                    \label{eq:energy-convergence}
\end{align}
The bound is uniform over each fixed finite collection of count
sectors.

We recall the elementary normalization estimate that makes this
sufficient. If $E,E'$ are two bounded energies on the same finite
reference-measure space and $\|E-E'\|_\infty\le\delta$, then
\[
 e^{-\beta\delta}\le Z'/Z\le e^{\beta\delta},\qquad
 e^{-2\beta\delta}
  \le\frac{e^{-\beta E'}/Z'}{e^{-\beta E}/Z}
  \le e^{2\beta\delta}.
\]
Their normalized probability laws therefore satisfy
\begin{equation}\label{eq:TV}
 \|\gamma'-\gamma\|_{\TV}
   \le\min\{1,e^{2\beta\delta}-1\},
\end{equation}
where total variation is the supremum over measurable events.
If $\delta\to0$ in probability, the right side tends to zero in
expectation because it is bounded by one. In particular no
exponential moment of the random full energy is required.

Let $F$ be a bounded continuous local function on $\Conf(\R)$.
Allowing multiplicities in this ambient space ensures that $F$
is also defined at auxiliary tuples with coincident coordinates.
Write
\[
 T_\Lambda(\omega,x)
       =\omega|_{\Lambda^c}+\sum_{i=1}^k\delta_{x_i}.
\]
For a fixed limiting sample and fixed $k$, the addition map is
continuous, so~\eqref{eq:count-stability} implies
\begin{equation}\label{eq:replacement-continuity}
 \sup_{x\in\overline\Lambda^k}
 |F(T_\Lambda(\xi_N,x))-F(T_\Lambda(\xi,x))|
                         \longrightarrow0.
\end{equation}
One can see the uniformity by taking the product of the compact
closure of the convergent exterior sequence with the compact cube
$\overline\Lambda^k$. Its image under addition is compact, and
$F$ is uniformly continuous on this image.

Denote the circle density from~\eqref{eq:finite-energy} by
$\gamma_{N,\Lambda}$, acting here by the local replacement
$T_\Lambda$. On each fixed count sector, split the difference
of its expectation of $F$ from $\gamma_\Lambda F$ into the
density difference and the replacement-observable difference.
Equations~\eqref{eq:energy-convergence}--\eqref{eq:replacement-continuity}
show that this difference tends to zero in probability. The
expectations of $F$ are bounded by $\|F\|_\infty$, so convergence
also holds in $L^1$. Finally the counts are tight by
\eqref{eq:count-stability}; restricting first to $k\le m$ and
then sending $m\to\infty$ yields
\begin{equation}\label{eq:operator-convergence}
 \gamma_{N,\Lambda}F(\xi_N)
        \longrightarrow\gamma_\Lambda F(\xi)
        \quad\text{in }L^1.
\end{equation}

\subsection{The finite-circle identity and canonical conditioning}

For a fixed local observable $F$, take $N$ sufficiently large
that both its support and $\overline\Lambda$ lie strictly inside
$C_N$. The ordinary finite-dimensional Gibbs conditional law
given the circle particles in $C_N\setminus\Lambda$ and the
inside count has exactly density~\eqref{eq:finite-energy}.
Resampling in the circle changes all the periodically repeated
copies when the configuration is lifted to the line. All changed
copies other than those in $\Lambda$ are outside the fixed
support of $F$ for sufficiently large $N$. Thus the circle Gibbs
identity is exactly
\begin{equation}\label{eq:circle-invariance}
 \E_{\mu_N}F=\E_{\mu_N}\gamma_{N,\Lambda}F
\end{equation}
for the local replacement operator used above. This identity
conditions inside one fundamental circle; the entire exterior
of an infinite periodic lift is not used as finite-volume
conditioning data.

Vague convergence passes the left side of
\eqref{eq:circle-invariance} to $\E_\mu F$, and
\eqref{eq:operator-convergence} passes the right side to
$\E_\mu\gamma_\Lambda F$. Hence
\[
 \mu F=\mu\gamma_\Lambda F
\]
for every bounded continuous local $F$. In particular it holds
for the Laplace cylinders
$F_f(\omega)=\exp(-\int f\,d\omega)$ with
$f\ge0$ continuous and compactly supported. Their finite linear
combinations form a multiplicative algebra generating the vague
Borel sigma-field. The probability measures $\mu$ and
$\mu\gamma_\Lambda$ therefore agree on all Borel sets, by the
functional monotone-class theorem.

We can now use the properness and measurability of the limiting
kernel. For every bounded Borel $F$ and every bounded
$\mathcal C_\Lambda$-measurable $H$,
\[
 \E_\mu[FH]
   =\E_\mu[\gamma_\Lambda(FH)]
   =\E_\mu[H\gamma_\Lambda F],
\]
by~\eqref{eq:properness}. Since $\gamma_\Lambda F$ itself is
$\mathcal C_\Lambda$-measurable, this proves
\begin{equation}\label{eq:DLR-conditional}
 \E_\mu[F\mid\xi|_{\Lambda^c},\xi(\Lambda)]
                    =\gamma_\Lambda F\quad\text{almost surely}.
\end{equation}
The density of this kernel is~\eqref{eq:gibbs}. Its ordinary
spatial principal value, independence of the reference point and
cutoff center, and consistency were proved in
Section~\ref{sec:intrinsic}. The further matched extraction began
from an arbitrary selected periodic weak limit, so this proves
Theorem~\ref{thm:DLR} for every such limit.

\section{Further questions}

The finite-volume displacement estimate and the canonical Gibbs
property leave several questions concerning the organization of
the infinite-volume states. One is whether periodic and
interval-background approximants select the same stationary
limits. This requires an ensemble-comparison argument beyond the
canonical Gibbs equations. Another is whether the periodic laws
converge along the full sequence and how their ergodic components
are classified.

The reciprocal-lattice estimate proves a positive limiting second
moment at low temperature. Determining its behavior throughout
the temperature range requires information beyond the displacement
second moment. The resistance comparison gives a uniform bound
because $|\theta|^{a-1}$ is integrable at the origin for $a>0$;
as $a$ approaches zero, this mechanism reaches its borderline.
Quantifying that endpoint behavior may help relate the present
positional bounds to fluctuations in the logarithmic gas.

\end{document}